\documentclass[12pt]{article}

\usepackage{amsmath,amsthm,amssymb,amscd,epsfig,url}
\usepackage{fullpage}
\usepackage{times}
\newtheorem{theorem}{Theorem}[section]
\newtheorem{lemma}[theorem]{Lemma}
\newtheorem{proposition}[theorem]{Proposition}
\newtheorem{corollary}[theorem]{Corollary}

\newtheorem{definition}[theorem]{Definition}

\newtheorem{remark}[theorem]{Remark}

\def\<{{\langle}} \def\>{{\rangle}}       
\title{A Note on Property Testing Sum of Squares and Multivariate Polynomial Interpolation}
\author{
  Aaron Potechin\thanks{Institute for Advanced Study.
    Supported by the Simons Collaboration for Algorithms and Geometry and by the NSF under agreement No. CCF-1412958.
    Part of this work was done while at Cornell University. }
  \and
  Liu Yang\thanks{Yale University}}
\date{\today}

\begin{document}
\maketitle

\begin{abstract}
In this paper, we investigate property testing whether or not a degree d multivariate polynomial is a sum of squares or is far from a sum of squares. We show that if we require that the property tester always accepts YES instances and uses random samples, $n^{\Omega(d)}$ samples are required, which is not much fewer than it would take to completely determine the polynomial. To prove this lower bound, we show that with high probability, multivariate polynomial interpolation matches arbitrary values on random points and the resulting polynomial has small norm. We then consider a particular polynomial which is non-negative yet not a sum of squares and use pseudo-expectation values to prove it is far from being a sum of squares.
\end{abstract}

\thispagestyle{empty}
.\newpage

\section{Introduction}
In recent years, property testing and the sum of squares hierarchy have both been fruitful areas of research. In property testing, we aim to find algorithms which only look at a small portion of the input. However, instead of requiring an exact answer, we only require that we can distinguish between a function which has a given property and a function which is far from having that property. Thus far, property testers have been found for many properties of boolean functions including monotonicity, dictatorships, juntas, and being low degree \cite{monotonicity,ParnasRS02,Blaisjuntas}. For a survey on results in property testing, see Oded Goldreich's book \cite{Goldreichbook}.

The sum of squares hierarchy, independently investigated by Nesterov \cite{nesterov}, Shor \cite{Shor87}, Parrilo \cite{Parrilo00}, and Lasserre \cite{Lasserre01}, is a hierarchy of semidefinite programs which has the advantages of being broadly applicable, powerful, and in some sense, simple. The sum of squares hierarchy is broadly applicable because it can be applied to any system of polynomial equations over the reals and most problems of interest can be put into this form. The sum of squares hierarchy is surprisingly powerful; it captures the best known algorithms for several problems including the Goemans-Williamson algorithm for maximum cut \cite{GoemansW95}, the Geomans-Linial relaxation for sparsest cut (analyzed by Arora,Rao,Vazirani \cite{AroraRV09}), and the subexponential time algorithm found by Arora, Barak, and Steurer \cite{AroraBS10} for unique games. Finally, the sum of squares hierarchy is in some sense simple as all that it uses is the fact that squares must be non-negative over the real numbers. For a survey on the sum of squares hierarchy, see Barak and Steurer's survey \cite{bd}. 

A central question in researching the sum of squares hierarchy is determining whether a given polynomial is non-negative and whether it is a sum of squares. In the setting where we know all the coefficients of the polynomial, we can determine whther it is a sum of squares in polynomial time using semidefinite programming while determining whether it is non-negative is NP-hard. In this paper, we consider the question of property testing whether a polynomial is a sum of squares on random samples. In this setting, rather than knowing the full polynomial, we only have its value on randomly sampled points. However, we only need to determine whether it is a sum of squares or is far from being a sum of squares.

This work is also related to research on the difference between non-negative polynomials and polynomials which are sum of squares. This research began with Hilbert \cite{Hilbert}, who proved the existence of polynomials which are non-negative yet not a sum of squares. The first explicit example of such a polynomial was found by Motzkin \cite{Motzkin67}. More recently, Bleckherman \cite{Blekherman} showed that there are significantly more polynomials which are non-negative than polynomials which are sums of squares. That said, to the best of our knowledge these papers do not analyze the distance of these polynomials from being sums of squares.

\subsection{Results and Outline}
Our main result is the following
\begin{theorem}
For all $d \geq 2$ and all $\delta > 0$, there is an $\epsilon > 0$ such that for sufficiently large $n$, if we require that our property tester always accepts YES instances and use random samples then 
property testing whether a degree $2d$ polynomial is a sum of squares requires at least $n^{\frac{d}{2} - \delta}$ samples.
\end{theorem}
Along the way, we prove the following result for multivariate polynomial interpolation on random points:
\begin{theorem}
For all $d$ and all $\delta > 0$, given points $p_1,\cdots,p_m \in \mathbb{R}^n$ randomly sampled from the multivariate normal distribution with covariance matrix $Id$, if $n$ is sufficiently large and $m \leq n^{d-\delta}$ then with very high probability, for all $v_1,\cdots,v_m$, there is a polynomial $g$ of degree $d$ such that 
\begin{enumerate}
\item $\forall i, g(p_i) = v_i$
\item $||g||$ is $O(\frac{||v||}{n^{\frac{d}{2}}})$
\item If we further have that $m \leq n^{\frac{d}{2}-\delta}$ then $||g^2||$ is $\tilde{O}(\frac{||v||^4}{n^{2d}})$
\end{enumerate}
\end{theorem}
This paper is organized as follows. In Section \ref{prelims} we give definitions and conventions which we will use for the remainder of the paper. In Section \ref{nonnegativitysection}, as a warm-up we consider the question of property testing non-negativity. This question is non-trivial because of how distance is defined in our setting. In Section \ref{sostestersection} we describe our tester for being a sum of squares which we will prove a lower bound against. In Section \ref{interpolationsection}, we prove our theorem on multivariate polynomial interpolation, showing that this tester will accept with high probability as long as the values it receives are non-negative and not too large. Finally, in Section \ref{pseudoexpectationconstruction} we complete our lower bound by giving a non-negative function $f$ of norm $1$ and lower bounding its distance from being a sum of squares using pseudo-expectation values.

\section{Preliminaries}\label{prelims}
For our results, we consider randomly sampling bounded degree real-valued polynomials over the multivariate normal distribution. We use the following conventions
\begin{enumerate}
\item We take $d$ or $2d$ to be the degree of our polynomials and assume that $d$ is a constant.
\item We take $m$ to be the number of sampled points.
\item Often, we will not be precise with functions of $d$ or logarithmic factors, so we absorb such functions into an $\tilde{O}$.
\end{enumerate}
We use the following definitions on the multivariate normal distribution.
\begin{definition}\label{normaldistdef} \ 
\begin{enumerate}
\item $\mathcal{N}(0,1)$ is the univariate normal distribution with probability density $\mu(x) = \frac{1}{\sqrt{2\pi}}e^{-\frac{x^2}{2}}$.
\item $\mathcal{N}(0,Id)$ is the multivariate normal distribution with probability density 
\[
\mu(x_1,\cdots,x_n) = \frac{1}{(2\pi)^{\frac{n}{2}}}e^{-\frac{||x||^2}{2}} = \prod_{i=1}^{n}{\mu(x_i)}
\]
\item For real-valued functions $f,g$, we define the inner product $\langle{f,g}\rangle = \int_{\mathcal{N}(0,Id)}{fg}$. We define $||f|| = \sqrt{\langle{f,f}\rangle}$
\item Given a set of real valued functions $S$ and a function $f$, we define the distance of $f$ from $S$ to be $d(f,S) = \min_{g:g \in S}{\{||f-g||\}}$
\end{enumerate}
\end{definition}
\begin{remark}\label{distdefremark}
It should be noted that this definition of distance differs from the definition of distance commonly used in the property testing literature, which is $d(f,S) = \min_{g:g \in S}{\{\mu(\{x:f(x) \neq g(x)\})\}}$. We use this definition of distance as it is more suitable for analyzing polynomials; if two polynomials have almost identical coefficients they will be very close to each other under Definition \ref{normaldistdef} but will be distance $1$ from each other using this definition.
\end{remark}
To help us index monomials, we use the following definitions:
\begin{definition} \ 
\begin{enumerate}
\item We use $I$ (and occasionally $J$) to denote a multi-set of elements in $[1,n]$.
\item We define $x_I = \prod_{i \in I}{x_i}$.
\item We define $|I|$ to be the total number of elements of $I$ (counting multiplicities)
\item Given an $I$ and a $k \in [1,n]$, we define $I_k$ to be the multiplicity of $k$ in $I$.
\item Given an $I$ and a $t \in [1,|I|]$, we define $I(t)$ to be the number such that $\sum_{j=1}^{I(t)-1}{I_j} < t$ but 
$\sum_{j=1}^{I(t)}{i_j} \geq t$. In other words, if we put the elements of $I$ in sorted order, $I(t)$ will be the tth element which appears.
\end{enumerate}
\end{definition}
\begin{remark}
Sometimes we will also attach subscripts to $I$. To distinguish between this and the notation above, we will only use the above notation with the letters $k$ and $t$ and in the case where both occur, we will put the subscript in parentheses on the inside. For example, if we want the multiplicity of $k$ in $I_j$ then we will write $(I_j)_k$
\end{remark}
For our analysis, it will be extremely useful to work with the orthonormal basis of polynomials. For the multivariate normal distribution, this basis is the Hermite polynomials and we use the following definitions
\begin{definition} \ 
\begin{enumerate}
\item We define $h_j(x)$ to be the jth Hermite polynomial normalized so that $||h_j|| = 1$
\item Given an $I$, we define $h_{I}(x_1,\cdots,x_n) = \prod_{k=1}^{n}{h_{I_k}(x_k)}$
\end{enumerate}
\end{definition}
\begin{proposition}
The multivariate polynomials $h_I(x_1,\cdots,x_n)$ are an orthonormal basis over $\mathcal{N}(0,Id)$.
\end{proposition}
\begin{corollary}[Fourier decomposition and Parseval's theorem]
For any polynomial $g$ we can write $g = \sum_{I}{{c_I}{h_I}}$ where $c_I = \langle{g,h_I}\rangle$ and we have that $||g||^2 = \sum_{I}{c^2_I}$.
\end{corollary}
\begin{remark}
Our results can be generalized to different product distributions. To do this, simply replace the Hermite polynomials with the appropriate orthonormal basis in a single variable.
\end{remark}
Finally, we need the concept of pseudo-expectation values, which is extremely useful for analyzing the sum of squares hierarchy. As we show below, pseudo-expectation values allow us to lower bound the distance of a degree $2d$ polynomial $f$ from being a sum of squares.
\begin{definition}
We define degree $2d$ pseudo-expectation values to be a linear map $\tilde{E}$ from polynomials of degree at most $2d$ to $\mathbb{R}$ which satisfies the following conditions:
\begin{enumerate}
\item $\tilde{E}[1] = 1$
\item $\forall g: deg(g) \leq d, \tilde{E}[g^2] \geq 0$
\end{enumerate}
\end{definition}
This second condition can be equivalently stated as follows
\begin{definition}
Given degree $2d$ pseudo-expectation values $\tilde{E}$, define the moment matrix $M$ to be the matrix with rows and columns indexed by monomials $\{x_I: |I| \leq d\}$ and entries $M_{IJ} = \tilde{E}[{x_I}x_J]$
\end{definition}
\begin{proposition}
The condition that $\forall g: deg(g) \leq d, \tilde{E}[g^2] \geq 0$ is equivalent to the condition that $M \succeq 0$.
\end{proposition}
We now show how pseudo-expectation values can be used to show a lower bound on how far a polynomial $f$ is from being a sum of squares of degree $\leq 2d$.
\begin{lemma}\label{distlowerbound}
Given pseudo-expectation values $\tilde{E}$, if $\tilde{E}[f] < 0$ then for all $g$ of degree at most $2d$ such that $g$ is a sum of squares, $||f-g||^2 \geq \frac{(\tilde{E}[f])^2}{\sum_{I:|I| \leq 2d}{(\tilde{E}[h_I])^2}}$
\end{lemma}
\begin{proof}
Write $g-f = \sum_{I}{{c_I}h_I}$. Observe that $\tilde{E}[g-f] = -\tilde{E}[f] + \tilde{E}[g] \geq -\tilde{E}[f]$ which implies that 
$$\sum_{I}{{c_I}\tilde{E}[h_I]} \geq -\tilde{E}[f]$$
Using Cauchy-Schwarz, 
$$\sqrt{\sum_{I}{c^2_I}}\sqrt{\sum_{I}{(\tilde{E}[h_I])^2}} \geq \sum_{I}{{c_I}\tilde{E}[h_I]} \geq -\tilde{E}[f]$$
Since both sides are non-negative, the result follows by squaring both sides and dividing both sides by $\sum_{I}{(\tilde{E}[h_I])^2}$.
\end{proof}
\section{Property testing non-negativity of degree $d$ polynomials}\label{nonnegativitysection}
As a warm-up, in this section we consider the closely related question of property testing whether a degree $d$ polynomial is non-negative or far from being non-negative. While this question is trivial under the definition of distance in Remark \ref{distdefremark}, it is non-trivial with our norm-based definition of distance. We also note that to the best of our knowledge, this problem is open if we consider the distance from the smaller set of non-negative degree $d$ polynomials rather than the set of all non-negative functions.
\begin{theorem}
The following property tester distinguishes with high probability between an $f$ which is a degree $d$ non-negative polynomial and an $f$ which is a degree $d$ polynomial that is $(\epsilon||f||)$-far from being non-negative.
\begin{enumerate}
\item Take $\frac{10B}{\epsilon}$ random samples where $B = e\left(4 + \ln{\left(\frac{1}{\epsilon}\right)}\right)^{2d}$
\item If any sample gives a negative value, return NO. Otherwise, return YES.
\end{enumerate}
\end{theorem}
\begin{proof}
Normalize $f$ so that $||f|| = 1$. Let $f^{-}$ be the negative part of $f$ and let $f^{+}$ be the non-negative part of $f$. If $f$ is $\epsilon$-far from being non-negative yet $\frac{10B}{\epsilon}$ random samples fails to find a negative value of $f$ with high probability then we must have that $||f^{-}|| > \epsilon$ yet $f^{-}$ is supported on a set of measure at most $\frac{\epsilon}{B}$.

However, by a corollary of the hypercontractivity theorem (which applies in the Gaussian setting as well, see O'Donnell's lecture notes on hypercontractivity \cite{LectureA}), for all $q$, 
\[
||f^{-}||_q \leq ||f||_q \leq (\sqrt{q-1})^d||f||_2 = (\sqrt{q-1})^d
\]
Given that $||f^{-}||_2 \geq \epsilon$ and $f^{-}$ is supported on a set of measure at most $\frac{\epsilon}{B}$, for $q > 2$ we minimize $||f^{-}||_q$ (over all functions, not just polynomials) by setting $f^{-}$ equal to $-\sqrt{B}$ on a set of measure $\frac{\epsilon}{B}$ and setting $f^{-} = 0$ elsewhere. This implies that
\[
{\epsilon}B^{\frac{q}{2}-1} \leq ||f^{-}||^q_q < q^{\frac{dq}{2}}
\]
This gives a contradition when $B \geq \left(\frac{1}{\epsilon}q^{\frac{dq}{2}}\right)^{\frac{2}{q-2}}$. Taking $q = 4 + \ln{(\frac{1}{\epsilon})}$,
\[
\left(\frac{1}{\epsilon}q^{\frac{dq}{2}}\right)^{\frac{2}{q-2}} \leq \left(\frac{1}{\epsilon}\right)^{\frac{1}{1 + \ln{(\frac{1}{\epsilon})}}}\left(4 + \ln{\left(\frac{1}{\epsilon}\right)}\right)^{2d}
\leq e\left(4 + \ln{\left(\frac{1}{\epsilon}\right)}\right)^{2d}
\]
Thus, we have a contradiction as long as $B \geq e\left(4 + \ln{\left(\frac{1}{\epsilon}\right)}\right)^{2d}$.
\end{proof}
If we instead consider the distance from non-negative degree $d$ polynomials, it is no longer clear whether any degree $d$ $f$ which is far from being a non-negative degree $d$ polynomial must be negative on a constant proportion of the inputs. We leave this as a question for further research. 

\section{Algorithm for testing SOS}\label{sostestersection}
In this section, we describe a tester for property testing whether a polynomial $f$ of degree $2d$ is a sum of squares of norm at most $1$ or is far from being a degree $2d$ sum of squares of norm at most $1$. This tester is optimal over all testers which always accept YES instances. Thus, to prove our lower bound it is sufficient to show that this tester fails with high probability.
 
Given data $\{f(p_i) = v_i, i \in [1,m]\}$, we can try to test whether a polynomial $f$ of degree at most $2d$ is a sum of squares as follows.
\begin{definition}
Given a coefficient matrix $M$ with rows and columns indexed by multi-sets $I$ of size at most $d$, define $f_{M} = \sum_{J}{\left(\sum_{I,I': I \cup I' = J}{M_{II'}}\right)x_J}$.
\end{definition}
\begin{proposition}
A polynomial $f$ can be written as a sum of squares if and only there exists a coefficient matrix $M$ such that $f_M = f$ and $M \succeq 0$.
\end{proposition}
Thus, we can search for a coefficient matrix $M$ such that 
\begin{enumerate}
\item $\forall i \in [1,m], f_M(x_i) = v_i$
\item $M \succeq 0$
\end{enumerate}
If such a coefficient matrix $M$ is found then we output YES, otherwise we output NO. 

This algorithm outputs YES precisely when there is polynomial $f_M$ of degree at most $2d$ which is a sum of squares and matches the data. However, for all we know, $||f_M||$ could be very high. On the other hand, in multivariate polynomial interpolation, when the polynomial is underdetermined it is natural to minimize the norm of the polynomial. To take this into account, we instead consider the following property testing problem and algorithm: \newline
Assumption: One of the following cases holds:
\begin{enumerate}
\item $f$ has degree at most $2d$, $f$ is a sum of squares, and $||f|| \leq 1$.
\item For all $g$ such that $g$ has degree at most $2d$, $g$ is a sum of squares, and $||g|| \leq 1$, $||f-g|| > \epsilon$.
\end{enumerate}
Algorithm: Search for a coefficient matrix satisfying the following conditions:
\begin{enumerate}
\item $\forall i \in [1,m], f_M(x_i) = v_i$
\item $||f_M|| \leq 1$
\item $M \succeq 0$
\end{enumerate}
\begin{remark}
These conditions on $M$ are all convex, so this algorithm can be implemented with convex optimization.
\end{remark}
\begin{remark}
This algorithm is optimal if we require that the property tester always accept YES instances, as it says YES precisely when there is a function $f_M$ which is a sum of squares, matches the data, and has norm at most $1$.
\end{remark}
To prove our lower bound, it is necessary and sufficient to find a degree $2d$ polynomial $f$ of norm $1$ which is $(\epsilon)$-far from being a degree $2d$ sum of squares such that if we take $m$ randomly sampled points where $m \leq n^{\frac{d}{2} - \delta}$, this tester accepts $f$ with high probability.

\section{Norm Bounds for Multivariate Polynomial Interpolation}\label{interpolationsection}
In polynomial interpolation, we are given points $p_1,\cdots,p_m$ and values $v_1,\cdots,v_m$ and we want to find a polynomial $g$ of a given degree $d$ such that $\forall i, g(p_i) = v_i$. Single variable polynomial interpolation is very well understood; it can be achieved preciasely when $m \leq d+1$. However, multivariable polynomial interpolation is much less well understood. In this section, we consider the case when the $p_i$ are random. In this case, interpolation is almost surely possible as long as $m \leq \sum_{i=0}^{d}{\binom{n+i-1}{i}}$, where $n$ is the number of variables. However, this does not say anything about the norm of the resulting polynomial. In this section, we show that for all $\delta > 0$, if $m \leq n^{d - \delta}$ and $n$ is sufficiently large the we can find a $g$ which matches all the data and has small norm. Moreover, if $m \leq n^{\frac{d}{2} - \delta}$ then $||g^2||$ has small norm as well. More precisely, we show the following theorem.
\begin{definition}
We define $C$ to be the expected value of $\sum_{I:0 < |I| \leq d}{h_I(p)^2}$ for a random point $p$.
\end{definition}
\begin{theorem}\label{polynomialinterpolationtheorem}
For all $d$ and all $\delta > 0$, given points $p_1,\cdots,p_m \in \mathbb{R}^n$ randomly sampled from the multivariate normal distribution with covariance matrix $Id$, if $n$ is sufficiently large and $m \leq n^{d-\delta}$ then with very high probability, for all $v_1,\cdots,v_m$, there is a polynomial $g$ of degree $d$ such that 
\begin{enumerate}
\item $\forall i, g(p_i) = v_i$
\item $||g|| = (1 \pm o(1))\frac{||v||}{\sqrt{C}}$
\item If we further have that $m \leq n^{\frac{d}{2}-\delta}$ then $||g^2||$ is $\tilde{O}(\frac{||v||^4}{n^{2d}})$
\end{enumerate}
\end{theorem}
This theorem shows that our tester will accept with high probability as long as $m \leq n^{\frac{d}{2} - \delta}$ and all our sampled points have non-negative values. To see this, note that given data $f(p_i) = v_i$, this theorem says that with high probability there is a $g$ such that $g(p_i) = \sqrt{v_i}$ and $||g^2||$ has small norm. Thus, $g^2$ matches the data and has small norm so the tester must accept.
\subsection{Construction of the function $g$}
To construct our function $g$, we use the following strategy:
\begin{enumerate}
\item We construct a function $g_i$ of degree $d$ for each point $p_i$.
\item We take the matrix $M$ where $M_{ij} = g_j(p_i)$.
\item We take $x$ to be a solution to $Mx = v$.
\item We take $g = \sum_{j}{{x_j}g_j}$.
\end{enumerate}
\begin{proposition}
For all $i$, $g(p_i) = v_i$.
\end{proposition}
\begin{proof}
For all $i$, $g(p_i) = \sum_{j}{{x_j}g_j(p_i)} = \sum_{j}{M_{ij}x_j} = v_i$.
\end{proof}
We now construct the functions $g_i$ of degree $d$ for each point $p_i$. These functions are constructed so that with high probability, for all $i$, $g_i(p_i) \approx 1$ and for all $i \neq j$, $|g_i(p_j)|$ is small. 
\begin{definition}
Given a point $p_i = (v_1,\cdots,v_n)$, we define $g_i = \frac{\sum_{I:0 < |I| \leq d}{h_I(p_i)h_I}}{C}$
\end{definition}
\subsection{Analysis of the function $g$}
To analyze the function $g$, it is useful to consider the following matrix $H$ which is closely related to $M$.
\begin{definition}
We define $H$ to be the matrix with rows indexed by $I$ where $0 < |I| \leq d$, columns indexed by $i$, and entries $H_{Ii} = h_I(p_i)$.
\end{definition}
\begin{lemma}
$M = \frac{{H^T}H}{C}$
\end{lemma}
\begin{proof}
Observe that 
\[
\frac{1}{C}({H^T}H)_{ij} = \frac{1}{C}\sum_{I:0 <|I| \leq d}{(h_I(p_i)h_I)(p_j)} = \left(\frac{\sum_{I:0<|I| \leq d}{h_I(p_j)h_I}}{C}\right)(p_i) = g_j(p_i) = M_{ij}
\]
\end{proof}
\begin{proposition}
$g = \frac{1}{C}\sum_{I}{{(Hv)_I}h_I}$
\end{proposition}
We now have that 
\[
||g||^2 = \frac{1}{C^2}\sum_{I}{\left(\sum_{j=1}^{m}{{x_j}H_{Ij}}\right)^2} = \frac{1}{C^2}{x^T{H^TH}x} = \frac{1}{C}{x^T}Mx = \frac{1}{C}{v^T}M^{-1}v
\]
In the next subsection, we will show that with high probability $M$ is very close to the identity which immediately implies that with high probability,  $||g|| = (1 \pm o(1))\frac{||v||}{\sqrt{C}}$, as needed.
\subsection{Analysis of $H$ and $M$}
In this subsection, we analyze the matrices $H$ and $M$. We begin by analyzing $H$ in order to develop the necessary techniques.
\begin{theorem}\label{Hnormboundtheorem}
For all $\delta > 0$ and all $d$, if $n$ is sufficiently large and $m \leq n^{d - \epsilon}$ then with high probability, $||H||$ is $\tilde{O}(n^{\frac{d}{2}})$
\end{theorem}
\begin{proof}
We can use the trace power method to probabilistically bound $||H||$. For this, we need to bound 
\[
E\left[tr((HH^T)^q)\right] = \sum_{i_1,\cdots,i_q,I_1,\cdots,I_q: \forall j, 0 < |I_j| \leq d}{E\left[\prod_{j=1}^{q}{H_{I_j{i_j}}H_{I_{j+1}{i_j}}}\right]}
\]
where we take $I_{q+1} = I_1$ and $i_{q+1} = i_1$. We partition this sum based on the intersection pattern $P$ of which of the $i_1,\cdots,i_q$ are equal to each other and how $I_1,\cdots,I_q$ interact with each other. We then analyze which intersection patterns give terms with nonzero expected value.
\begin{definition}
We define an intersection pattern $P$ to be the following data:
\begin{enumerate}
\item For all $j' \neq j$, $P$ has the equality $i_{j'} = i_j$ or the inequality $i_{j'} \neq i_j$
\item For all $j,j',t,t'$, $P$ has the equality $I_{j'}(t') = I_j(t)$ or the inequality $I_{j'}(t') \neq I_j(t)$
\end{enumerate}
where these equalities and inequalities are consistent with each other (i.e. transitivity is satisfied for the equalities).
\end{definition}
\begin{lemma}\label{countingintersectionpatterns}
There are at most $(4dq)^{4dq}$ possible intersection patterns.
\end{lemma}
\begin{proof}
Choose an arbitrary ordering of the $i_j$ and an arbitrary ordering of the $I_j(t)$. To specify an intersection pattern, it suffices to specify which $i_j$ and $I_j(t)$ are equal to previous $i_j$ and $I_j(t)$ and if so, to specify one of the equalities which hold. The total number of choices is at most $(4dq)^{4dq}$.
\end{proof}
\begin{lemma}\label{xylemma}
For any intersection pattern $P$ which gives a nonzero expected value, letting $x = |\{k:\exists j: (I_{j})_k > 0\}|$ and letting $y$ be the number of distinct $i_j$, $y + \frac{x}{d} \leq q+1$
\end{lemma}
\begin{proof}
The key observation is that if we consider the multiset 
\[
\{(i_j,k): (I_j)_k > 0\} \cup \{(i_j,k): (I_{j+1})_k > 0\} = \{(i_j,k): (I_j)_k > 0\} \cup \{(i_{j-1},k): (I_j)_k > 0\},
\]
if any element of this multiset appears exactly once then $E\left[\prod_{j=1}^{q}{H_{I_j{i_j}}H_{I_{j+1}{i_j}}}\right] = 0$ over the random choices for the points $\{p_i\}$.

With this observation in mind, for each $k$, consider the graph fomed by the edges $\{(i_{j-1},i_j): (I_j)_k > 0\}$. In a term with nonzero expectation, for all $k$, every vertex in this graph with nonzero degree must have degree at least 2 (where we consider loops as adding 2 to the degree). Thus, these graphs must consist of loops/cycles and loops/cycles joined by paths. This implies the following upper bound on $x$
\begin{definition}
Let $G_y$ be the multi-graph consisting of the $q$ edges $\{i_{j-1},i_j\}$
\end{definition}
\begin{definition}
Given a multi-graph $G$, we define $w(G)$ to be the maximum number such that $\exists G_1,\cdots,G_t$ and $w_1,\cdots,w_t$ satisfying the following conditions
\begin{enumerate}
\item $w(G) =\sum_{i}{w_i}$ 
\item $\forall i, V(G_i) = V(G), E(G_i) \subseteq E(G)$, $E(G_i)$ is nonempty, and no vertex of $G_i$ has degree exactly $1$ (where we consider loops as adding 2 to the degree).
\item $\forall i, w_i \geq 0$
\item $\forall j, \sum_{i: (i_{j-1},i_j) \in E(G_i)}{w_i} \leq 1$
\end{enumerate}
\end{definition}
\begin{lemma}\label{weightupperboundlemma}
For any intersection pattern which gives a nonzero expected value, $x = |\{k:\exists j: (I_{j})_k > 0\}| \leq d \cdot w(G_y)$
\end{lemma}
\begin{proof}
For each $k:\exists j: (I_{j})_k > 0$, we construct the graph $G_k$ where $V(G_k) = V(G_y),E(G_k) = \{(i_{j-1},i_j): (I_{j})_k > 0\}$ and assign it weight $\frac{1}{d}$. From the above observation, no vertex of $G_k$ can have degree exactly $1$ (where we consider loops as adding 2 to the degree). Also, we have that the total weight on any edge $(i_{j-1},i_j)$ is at most $1$ as at most $d$ graphs $G_k$ contribute to it and each contribution is $\frac{1}{d}$. Thus, 
$\sum_{k:\exists j: (I_{j})_k > 0}{\frac{1}{d}} = \frac{x}{d} \leq w(G_y)$, as needed.
\end{proof}
With this bound in mind, we now prove the following lemma which will immediately imply our result.
\begin{lemma}\label{weightvertextradeofflemma}
For all connected multi-graphs $G$, $w(G) + |V(G)| \leq |E(G)|+1$
\end{lemma}
\begin{proof}
We first reduce to the case where every non-loop edge of $G$ has multiplicity at least two with the following lemma.
\begin{lemma}
If $G$ is a multigraph which has a non-loop edge $e$ that appears with multiplicity $1$ and $G'$ is the graph formed by contracting this edge then $w(G) \leq w(G')$
\end{lemma}
\begin{proof}
Observe that if subgraphs $G_1,\cdots,G_t$ of $G$ all have no vertex of degree exactly $1$, then letting $G'_1,\cdots,G'_t$ be the graphs $G_1,\cdots,G_t$ formed by making the two endpoints of $e$ equal and removing $e$ (if present), $G'_1,\cdots,G'_t$ are subgraphs of $G'$ and have no vertices of degree exactly $1$. To see this, note that for any vertex $v'$ in $G'_i$ except for the vertex formed by making the two endpoints of $e$ equal, the number of edges incident to $v'$ is unaffected. For the $v'$ formed by making the two endpoints of $e$ equal, each of these endpoints must have had an edge besides $e$ incident with it, so the degree of this $v'$ is at least 2.
\end{proof}
Using this lemma, if $G$ has a non-loop edge $e$ which appears with multiplicity $1$, $G'$ is the graph formed by contracting this edge, and $w(G') + |V(G')| \leq |E(G')|+1$ then 
\[
w(G) + |V(G)| \leq w(G') + |V(G')| + 1 \leq |E(G')|+2 = |E(G)| + 1
\]
Thus, it is sufficient to prove the lemma for $G'$. Applying this logic repeatedly, it is sufficient to prove the result for the case where every non-loop edge of $G$ has multiplicity at least two.
\begin{definition}
We define $E_{loop}(G)$ to be the multi-set of loops in $G$ and we define $E_{nonloop}(G)$ to be $E(G) \setminus E_{loop}(G)$.
\end{definition}
\begin{lemma}
For all $G$, $w(G) \leq |E_{loop}(G)| + \frac{|E_{nonloop}(G)|}{2}$
\end{lemma}
\begin{proof}
Let $G_1,\cdots,G_k$ and $w_1,\cdots,w_k$ be graphs and weights such that 
\begin{enumerate}
\item $w(G) =\sum_{i=1}^{k}{w_i}$
\item $\forall i, V(G_i) = V(G), E(G_i) \subseteq E(G)$, $E(G_i)$ is nonempty, and no vertex of $G_i$ has degree exactly $1$ (where we consider loops as adding 2 to the degree).
\item $\forall i, w_i \geq 0$
\item $\forall j, \sum_{i \in [1,k]: (i_{j-1},i_j) \in E(G_i)}{w_i} \leq 1$
\end{enumerate}
Observe that each $G_i$ must either have at least one loop or at least two non-loop edges. Thus, 
\begin{align*}
|E_{loop}(G)| + \frac{|E_{nonloop}(G)|}{2} &\geq \sum_{j:(i_{j-1},i_j) \in E_{loop}(G)}{\sum_{i:(i_{j-1},i_j) \in E(G_i)}{w_i}} + 
\frac{1}{2}\sum_{j:(i_{j-1},i_j) \in E_{nonloop}(G)}{\sum_{i:(i_{j-1},i_j) \in E(G_i)}{w_i}} \\
&= \sum_{i}{\left(\sum_{j:(i_{j-1},i_j) \in E(G_i) \cap E_{loop}(G)}{w_i} + \frac{1}{2}\sum_{j:(i_{j-1},i_j) \in E(G_i) \cap E_{nonloop}(G)}{w_i}\right)} \\
&\geq \sum_{i}{w_i} = w(G)
\end{align*}
as needed.
\end{proof}
\begin{lemma}
If $G$ is connected and every non-loop edge of $G$ has multiplicity at least 2 then $|V(G)| \leq \frac{|E_{nonloop}(G)|}{2} + 1$
\end{lemma}
\begin{proof}
Imagine building up $G$ from one isolated vertex. We can add loops for free, but every time we add a neighbor of an existing vertex, we must add at least two edges (as all edges have multiplicity at least two).
\end{proof}
Putting these lemmas together, $|V(G)| + w(G) \leq |E_{loop}(G)| + |E_{nonloop}(G)| + 1 = |E(G)| + 1$
which proves Lemma \ref{weightvertextradeofflemma}.
\end{proof}
Putting Lemmas \ref{weightupperboundlemma} and \ref{weightvertextradeofflemma} together, Lemma \ref{xylemma} follows immediately. We have that  
$y + \frac{x}{d} \leq |V(G)| + w(G) \leq |E(G)| + 1 = q+1$, as needed.
\end{proof}
We now consider the expression 
\[
E\left[tr((HH^T)^q)\right] = \sum_{i_1,\cdots,i_q,I_1,\cdots,I_q: \forall j, 0 < deg(I_j) \leq d}{E\left[\prod_{j=1}^{q}{H_{I_j{i_j}}H_{I_{j+1}{i_j}}}\right]}
\]
Lemma \ref{xylemma} implies that for any intersection pattern which gives a nonzero expected value, there are at most $\max_{x,y:y\geq 1,y+\frac{x}{d} \leq q+1}{\{{m^y}n^x\}} \leq mn^{dq}$ choices for $i_1,\cdots,i_q,I_1,\cdots,I_q$. By Lemma \ref{countingintersectionpatterns}, there are at most $(4dq)^{4dq}$ possible intersection patterns. To complete our upper bound, we just need to show a bound on $E\left[\prod_{j=1}^{q}{H_{I_j{i_j}}H_{I_{j+1}{i_j}}}\right]$ for a particular $i_1,\cdots,i_q,I_1,\cdots,I_q$, which we do with the following lemma
\begin{lemma}
For any $i_1,\cdots,i_q,I_1,\cdots,I_q$,
\[
E\left[\prod_{j=1}^{q}{H_{I_j{i_j}}H_{I_{j+1}{i_j}}}\right] \leq (4dq)^{4dq}
\] 
\end{lemma}
\begin{proof}
Observe that for any $i_1,\cdots,i_q,I_1,\cdots,I_q$, $E\left[\prod_{j=1}^{q}{H_{I_j{i_j}}H_{I_{j+1}{i_j}}}\right]$ is a product of expressions of the form 
$E\left[\left(\prod_{i=1}^{k}{h_{j_i}}\right)(x)\right]$ for some $j_1,\cdots,j_k$. For each such expression we have the following bound.
\begin{lemma}
Letting $d' = \sum_{i=1}^{k}{j_i}$, $E\left[\left(\prod_{i=1}^{k}{h_{j_i}}\right)(x)\right] \leq (d')^{2d'}$
\end{lemma}
\begin{proof}
We use the fact that for all $j \geq 1$, the sum of the absolute values of the coefficients of $h_j$ is at most $j^j$. This implies that the sum of the absolute values of the coefficients of $\prod_{i=1}^{k}{h_{j_i}}$ is at most ${d'}^{d'}$. Over a normal distribution $E[x^{p}] = \prod_{i=1}^{\frac{p}{2}}{(2i-1)} \leq p^p$ if $p$ is even and is $0$ if $p$ is odd, which implies the result.
\end{proof}
The total sum of all the degrees is at most $2dq$ as this is the maximum number of pairs $I_j(t),i_j$ and $I_{j+1}(t),i_j$. Thus, the product over all of the expressions which we have is at most $(4dq)^{4dq}$, as needed.
\end{proof}
Putting everything together, 
\[
E\left[tr((HH^T)^q)\right] = \sum_{i_1,\cdots,i_q,I_1,\cdots,I_q: \forall j, 0 < deg(I_j) \leq \frac{d}{2}}{E\left[\prod_{j=1}^{q}{H_{I_j{i_j}}H_{I_{j+1}{i_j}}}\right]} 
\leq (4dq)^{8dq}mn^{dq}
\]
We now apply Markov's inequality. For all $q$ and all $\beta \geq 0$, 
\begin{align*}
Pr\left[||H|| \geq \sqrt[2q]{n^{\beta}E\left[tr((HH^T)^q)\right]}\right] &= Pr\left[||H||^{2q} \geq n^{\beta}E\left[tr((HH^T)^q)\right]\right] \\
& \geq Pr\left[tr((HH^T)^q) \geq n^{\beta}E\left[tr((HH^T)^q)\right]\right] \leq \frac{1}{n^{\beta}}
\end{align*}
Applying this with $q \sim dq\beta\log{n}$, Theorem \ref{Hnormboundtheorem} follows.
\end{proof}
With the techniques we developed to prove Theorem \ref{Hnormboundtheorem}, we can now anaylze $M$.
\begin{theorem}\label{Manalysistheorem}
For all $d$ and all $\delta > 0$, for sufficiently large $n$, if $m \leq n^{d - \delta}$ and we write $M = Id + M'$ then with high probability $||M'|| << 1$.
\end{theorem}
\begin{proof}[Proof sketch:]
This theorem can be proved by considering the diagonal part and off-diagonal part of $M$. For the diagonal part of $M$, observe that $M_{ii} = \frac{1}{C}\sum_{I:0 <|I| \leq d}{h_I(p_i)^2}$. Since $C$ is the epected value of $\sum_{I:0 <|I| \leq d}{h_I(p)^2}$ for a random point $p$ and this value is tightly concentrated around its expectation, with high probability $M_{ii}$ will be $1 \pm o(1)$ for all $i$.
For the off-diagonal part of $M$, we can use the trace power method to bound its norm. Let $M''$ be the off-diagonal part of $M$.
\begin{lemma}
For all $d$ and all $\delta > 0$, for sufficiently large $n$, if $m \leq n^{d - \delta}$ then with high probability $||M''||$ is $\tilde{O}\left(\frac{\sqrt{m}}{n^{\frac{d}{2}}}\right)$
\end{lemma}
\begin{proof}[Proof sketch:]
Observe that 
\[
E\left[tr((M'')^q)\right] = \frac{1}{C^q}\sum_{i_1,\cdots,i_q,I_1,\cdots,I_q: \forall j, 0 < deg(I_j) \leq \frac{d}{2}, i_{j} \neq i_{j+1}}{E\left[\prod_{j=1}^{q}{H_{I_j{i_j}}H_{I_{j+1}{i_j}}}\right]}
\]
Up to the $\frac{1}{C^q}$ factor, this is the same expression we had for $E\left[tr((HH^T)^q)\right]$ except that since we are restricting ourselves to the off-diagonal part of $M$ we additionally have the constraint that $i_{j} \neq i_{j+1}$ for all $j$. This constraint implies that for any term with nonzero expected value, there is no $k$ such that $(I_j)_k > 0$ and $(I_{j'})_{k} = 0$ for all $j \neq j'$. This in turn implies that we only need to consider intersection patterns with $x \leq \frac{dq}{2}$, which means that the maximum number of choices for a given intersection pattern is at most 
$(m^{(\frac{q}{2}+1)}n^{\frac{dq}{2}})$ rather than $mn^{dq}$. Thus, our final bound on $E\left[tr((M'')^q)\right]$ will be 
\[
E\left[tr((M'')^q)\right]  \leq (4dq)^{8dq}\frac{1}{C^q}m^{(\frac{q}{2}+1)}n^{\frac{dq}{2}}
\]
Recalling that $C = \sum_{I:0 < |I| \leq d}{h_I(p)^2}$, $C$ is $\Theta(n^d)$ and the result can be shown in same the way as Theorem \ref{Hnormboundtheorem} using Markov's inequality (where we choose an even $q$).
\end{proof}
\end{proof}
\begin{remark}
In fact, our analysis of $M$ gives us improved norm bounds on $||H||$. In particular, with high probability $||H||$ is $\sqrt{C}(1 \pm o(1))$
\end{remark}
\subsection{Analysis of $||g^2||$}
In this subsection, we show how to probabilistically bound $||g^2||$.
\begin{definition}
We define the matrix $Q$ so that $Q_{J(I,I')}$ is the coefficient of $h_J$ in ${h_I}h_{I'}$.
\end{definition}
We have that
\[
g^2 = \frac{1}{C^2}\sum_{J}{\left(Q_{J(I,I')}\sum_{I,I'}{\sum_{j=1}^{m}{\sum_{j' = 1}^{m}{{x_j}H_{Ij}x_{j'}H_{I'j}}}}\right)h_J}
\]
Thus,
\begin{align*}
||g^2||^2 &= \frac{1}{C^4}\sum_{J}{\left(Q_{J(I,I')}\sum_{I,I'}{\sum_{j=1}^{m}{\sum_{j' = 1}^{m}{{x_j}H_{Ij}x_{j'}H_{I'j}}}}\right)^2} \\
&= \frac{1}{C^4}(x \otimes x)^T((H \otimes H)^TQ^TQ(H \otimes H))(x \otimes x)
\end{align*}
\begin{theorem}
For all $d,\delta$ and all sufficiently large $n$, if $m \leq n^{\frac{d}{2} - \delta}$ then with high probability, for all vectors $x$, $||Q(H \otimes H)(x \otimes x)||^2$ is $\tilde{O}(n^{2d}||x||^4)$
\end{theorem}
\begin{proof}[Proof sketch]
We break $Q(H \otimes H)$ into two parts.
\begin{enumerate}
\item Let $A$ be the matrix such that $A_{J(j,j')} = (Q(H \otimes H))_{J(j,j')}$ if $j' = j$ and is $0$ otherwise.
\item Let $R$ be the matrix such that $R_{J{j,j'}} = (Q(H \otimes H))_{J(j,j')}$ if $j' \neq j$ and is $0$ otherwise.
\end{enumerate}
For the first part, we observe that letting $A_{(j,j)}$ be the $(j,j)$ column of $A$,
\[
||A(x \otimes x)|| = ||\sum_{j}{{x^2_j}A_{jj}}|| \leq \left(\sum_{j}{x^2_j}\right)\max_{j}{\{||A_{jj}||\}} = ||x||^2\max_{j}{\{||A_{jj}||\}}
\]
Thus, $||A(x \otimes x)||^2 \leq \left(\max_{j}{\{||A_{jj}||\}}\right)^2||x||^4$ and it is sufficient to probabilistically bound $\max_{j}{\{||A_{jj}||\}}$
\begin{lemma}
With high probability, $\max_{j}{\{||A_{jj}||\}}$ is $\tilde{O}(n^d)$
\end{lemma}
\begin{proof}
Observe that for all $j$, 
\[
||A_{jj}||^2 = \sum_{J,I_1,I_2,I_3,I_4}{Q_{J(I_1,I_2)}Q_{J(I_3,I_4)}H_{{I_1}j}H_{{I_2}j}H_{{I_3}j}H_{{I_4}j}}
\]
The entries of $Q$ are $O(1)$ and with high probability the entries of $H$ are $\tilde{O}(1)$, so we just need to bound the number of $I_1,I_2,I_3,I_4$ which give a nonzero contribution. For this, 
observe that for any nonzero term,
\begin{enumerate}
\item $I_1 \Delta I_2 \subseteq J$ and $I_3 \Delta I_4 \subseteq J$ where $\Delta$ is the symmetric difference.
\item $J \subseteq I_1 \cup I_2$ and $J \subseteq I_3 \cup I_4$
\end{enumerate}
Together, these observations imply that there cannot be a $k$ such that precisely one of $(I_1)_k,(I_2)_k,(I_3)_k,(I_4)_k$ is nonzero. In turn, this implies that there are $O(n^{2d})$ choices for $I_1,I_2,I_3,I_4$ which give a nonzero contribution and the result follows.
\end{proof}
For the second part, we bound the norm of $R$.
\begin{lemma}\label{Rnormboundlemma}
For all $d,\delta$ and all sufficiently large $n$, if $m \leq n^{\frac{d}{2} - \delta}$ then with high probability, $||R||$ is $\tilde{O}(n^{d})$.
\end{lemma}
\begin{proof}[Proof sketch]
This can be shown using the trace power method. We have that
\begin{align*}
&E\left[(R^TR)^q\right] = \\
&\sum_{\{j_1,j'_1,J_1,I_{11},I_{12},I_{13},I_{14}, \cdots, j_q,j'_q,J_q,I_{q1},I_{q2},I_{q3},I_{q4}\}}{E\left[\prod_{a=1}^{q}
{Q_{J_a(I_{a1},I_{a2})}Q_{J_a(I_{a3},I_{a4})}H_{{I_{a1}}j_a}H_{{I_{a2}}j'_a}H_{{I_{a3}}j_{a+1}}H_{{I_{a4}}j'_{a+1}}}\right]}
\end{align*}
Similar to before, we can partition this sum into intersection patterns and consider which patterns have nonzero expectation.
\begin{definition}
We take $x$ to be the number of distinct $k$ such that $(I_{ai}) > 0$ for some $a,i$ and we take $y$ to be the number of distinct $j_a$ and $j'_a$.
\end{definition}
\begin{lemma}
For any term with nonzero expected value, $y + \frac{2x}{d} \leq 4q + 2$
\end{lemma}
\begin{proof}[Proof sketch]
In any term with nonzero expected value, following the same logic as before, for each block $I_{a1},I_{a2},I_{a3},I_{a4}$ there cannot be a $k$ such that precisely one of $(I_{a1})_k,(I_{a2})_k,(I_{a3})_k,(I_{a4})_k$ is nonzero. If every $k$ which appears in a block appears in at least two blocks then $x \leq d$. We trivially have that $y \leq 2q$ so the result holds in this case.

If there is a $k$ which appears in only one block then this implies an equality between $j_a = j_{a+1}$ or an equality $j'_{a} = j'_{a+1}$. Roughly speaking, each such equality allows $d$ additional values $k$ to only appear in one block, decreasing $y$ by $1$ but increasing $x$ by $\frac{d}{2}$. This leaves $y + \frac{2x}{d}$ unchanged. To see why we have the $+2$, consider the extreme case when all the $j_a$ are equal and all the $j'_a$ are equal. In this case $y = 2$ and we can have $x = 2dq$.
\end{proof}
With this lemma in hand, since $m \leq n^{\frac{d}{2} - \delta}$, for any intersection pattern which gives a nonzero expected value, the total number of choices for the $j_a,j'_a,I_{a1},I_{a2},I_{a3},I_{a4}$ is $O(m^2n^{dq})$. Lemma \ref{Rnormboundlemma} can now be shown using the same techniques used to prove Theorem \ref{Hnormboundtheorem}.
\end{proof}
Putting these results together, it follows that with high probability, for all vectors $x$, $||Q(H \otimes H)(x \otimes x)||^2$ is $\tilde{O}(n^{2d}||x||^4)$, as needed.
\begin{remark}
While Lemma \ref{Rnormboundlemma} is essentially tight, it should be possible to obtain the same bound on $||g^2||^2$ for $m \leq n^{d - \delta}$ if we can effectively use the fact that we are dealing with $(x \otimes x)$ rather than an arbitrary vector, just as we did for $A$. We leave this as a question for further research.
\end{remark}
\end{proof}
\section{A non-negative polynomial which is far from being a sum of squares}\label{pseudoexpectationconstruction}
In this section, we complete our lower bound by giving a non-negative polynomial $f$ and showing that $f$ is far from being a sum of squares.
\begin{definition}
We take $f = \frac{r}{2}x^{r+2}y^{r} + \frac{r}{2}x^{r}y^{r+2} - (r+1)x^{r}y^{r} + (1+c)$ where $c \geq 0$, $r \geq 2$ is even, and we take $d = 2r+2$
\end{definition}
\begin{remark}
This polynomial $f$ is a generalization of the Motzkin polynomial (which is the case $r = 2,c=0$).
\end{remark}
\begin{lemma}
$f \geq c$
\end{lemma}
\begin{proof}
Observe that $(x^{r+2}y^{r})^{\frac{r}{2r+2}}(x^{r}y^{r+2})^{\frac{r}{2r+2}}(1)^{\frac{2}{2r+2}} = x^{r}y^{r}$ and $\frac{r}{2r+2} + \frac{r}{2r+2} + \frac{2}{2r+2} = 1$. By the AM-GM inequality, 
$\frac{r}{2r+2}(x^{r+2}y^{r}) + \frac{r}{2r+2}(x^{r}y^{r+2}) + \frac{2}{2r+2} \geq x^{r}y^{r}$ and the result follows.
\end{proof}
\begin{theorem}
$f$ is $\left(\frac{1}{(d^3\sqrt[r]{2+c})^{2d^4}}\right)$-far from being SOS.
\end{theorem}
\begin{proof}
We take the following pseudo-expectation values. 
\begin{definition}
Take $k > 1$ and take $B = (kd^3)^{(3d^3)}$. We split up the pseudo-expectation values $\tilde{E}[{x^a}{y^b}]$ into cases as follows
\begin{enumerate}
\item If $a > b$ then we take $\tilde{E}[{x^a}{y^b}] = \frac{(kd^3)^{(a^2 + (a+b)^2)}}{(kd^3)^{2d^2}}B^{a - \frac{r+2}{r}b}$.
\item If $b > a$ then we take $\tilde{E}[{x^a}{y^b}] = \frac{(kd^3)^{(b^2 + (a+b)^2)}}{(kd^3)^{2d^2}}B^{a - \frac{r+2}{r}b}$.
\item For all $a > 0$ we take $\tilde{E}[{x^a}{y^a}] = \frac{4^{(a^2)}}{4^{(r^2)}}{k^a}$
\item We have $\tilde{E}[1] = 1$.
\item We take $\tilde{E}[p(x,y)q(\text{other variables})] = \tilde{E}[p(x,y)]E[q]$. This guarantees that $\tilde{E}[h_{I}] = 0$ whenever $I$ contains a variable besides $x$ and $y$.
\end{enumerate}
\end{definition}
\begin{proposition}
$\tilde{E}[f] = c+(r+1)(1-k^r)$
\end{proposition}
\begin{proof}
This follows immediately from the observations that $\tilde{E}[{x^{r+2}}{y^{r}}] = \tilde{E}[{x^{r}}{y^{r+2}}] = 1$ and $\tilde{E}[x^{r}y^{r}] = k^r$.
\end{proof}
We need to show that these pseudo-expectation values give a PSD moment matrix.
\begin{proposition}
For all $a,b$ such that $a+b \leq d$, 
$$\tilde{E}[{x^a}{y^b}] \geq \min{\left\{\frac{(kd^3)^{(a^2 + (a+b)^2)}}{(kd^3)^{2d^2}}B^{a - \frac{r+2}{r}b}, 
\frac{(kd^3)^{(b^2 + (a+b)^2)}}{(kd^3)^{2d^2}}B^{a - \frac{r+2}{r}b}\right\}}$$
\end{proposition}
\begin{lemma}
For all $a_1,b_1,a_2,b_2$ such that $a_1 + b_1 \leq d$, $a_1 + b_2 \leq d$, $a_1 \neq a_2$ or $b_1 \neq b_2$, and $a_1 \neq b_1$ or $a_2 \neq b_2$, 
$$\tilde{E}[x^{a_1 + a_2}y^{b_1 + b_2}] \leq \frac{1}{2d^2}\sqrt{\tilde{E}[x^{2a_1}y^{2b_1}]\tilde{E}[x^{2a_2}y^{2b_2}]}$$
\end{lemma}
\begin{proof}
If $a_1 + a_2 > b_1 + b_2$ then we have that 
$$\sqrt{\tilde{E}[x^{2a_1}y^{2b_1}]\tilde{E}[x^{2a_2}y^{2b_2}]} \geq \frac{(kd^3)^{(2a_1^2 + 2(a_1+b_1)^2 + 2a_2^2 + 2(a_2+b_2)^2)}}{(kd^3)^{2d^2}}B^{a_1+a_2 - \frac{r+2}{r}(b_1+b_2)}$$
Thus, 
$$\frac{\tilde{E}[x^{a_1 + a_2}y^{b_1 + b_2}]}{\sqrt{\tilde{E}[x^{2a_1}y^{2b_1}]\tilde{E}[x^{2a_2}y^{2b_2}]}} \leq 
(kd^3)^{-((a_1-a_2)^2 + (a_1+b_1 - a_2 - b_2)^2)}$$
Since we either have that $a_1 \neq a_2$ or $a_1 + b_1 \neq a_2 + b_2$, this is at most $\frac{1}{kd^3} \leq \frac{1}{2d^2}$.

If $b_1 + b_2 > a_1 + a_2$ then we can use a symmetrical argument. If $a_1 + a_2 = b_1 + b_2$ and $a_1 > b_1$ then 
$$\sqrt{\tilde{E}[x^{2a_1}y^{2b_1}]\tilde{E}[x^{2a_2}y^{2b_2}]} = \frac{(kd^3)^{(2a_1^2 + 2(a_1+b_1)^2 + 2b_2^2 + 2(a_2+b_2)^2)}}{(kd^3)^{2d^2}}B^{a_1+b_2 - \frac{r+2}{r}(b_1+a_2)}$$
Since $a_1 + b_2 \geq b_1 + a_2 + 2$ we can't have $b_1 = a_2 = r$, $a_1+b_2 - \frac{r+2}{r}(b_1+a_2) \geq \frac{1}{r}$. Since $B \geq (kd^3)^{(3d^3)}$, 
$$\frac{(kd^3)^{(2a_1^2 + 2(a_1+b_1)^2 + 2b_2^2 + 2(a_2+b_2)^2)}}{(kd^3)^{2d^2}}B^{a_1+b_2 - \frac{r+2}{r}(b_1+a_2)} \geq 
\frac{(kd^3)^{\frac{1}{r}(3d^3)}}{(kd^3)^{2d^2}} \geq (kd^3)^{d^2}$$
Thus, 
$$\frac{\tilde{E}[x^{a_1 + a_2}y^{b_1 + b_2}]}{\sqrt{\tilde{E}[x^{2a_1}y^{2b_1}]\tilde{E}[x^{2a_2}y^{2b_2}]}} \leq \frac{k^d}{(kd^3)^{d^2}} \leq \frac{1}{2d^2}$$
\end{proof}
\begin{lemma}
For all $a,b$ such that $a \neq b$, $a \leq r$, $b \leq r$, and $a+b \leq d$, $\tilde{E}[x^{a+b}y^{a+b}] \leq \frac{1}{4^{|a-b|}}\sqrt{\tilde{E}[x^{2a}y^{2a}]\tilde{E}[x^{2b}y^{2b}]}$
\end{lemma}
\begin{proof}
We have that $\sqrt{\tilde{E}[x^{2a}y^{2a}]\tilde{E}[x^{2b}y^{2b}]} \geq \frac{4^{2(a^2 + b^2)}}{4^{(r^2)}}{k^{a+b}}$ while 
$\tilde{E}[x^{a+b}y^{a+b}] = \frac{4^{(a+b)^2}}{4^{(r^2)}}{k^{a+b}}$. Thus, 
$$\frac{\tilde{E}[x^{a+b}y^{a+b}]}{\sqrt{\tilde{E}[x^{2a}y^{2a}]\tilde{E}[x^{2b}y^{2b}]}} \leq \frac{1}{4^{(a-b)^2}} \leq \frac{1}{4^{|a-b|}}$$
\end{proof}
Combining all of these results, it can be shown that the moment matrix $M$ corresponding to $\tilde{E}$ is PSD. We now bound 
$\sum_{I:|I| \leq d}{(\tilde{E}[h_I])^2}$. By a large margin, the dominant terms will come from the leading coefficients of the degree $d$ Hermite polynomials for $x$ and $y$. The leading coefficient of the degree d Hermite polynomial is $\frac{1}{\sqrt{d!}}$ so we have that 
$$\sum_{I:|I| \leq d}{(\tilde{E}[h_I])^2} \leq \frac{2}{\sqrt{d!}}B^{\frac{d}{2}} \leq (kd^3)^{2d^4}$$
By Lemma \ref{distlowerbound}, for all $g$ of degree at most $d$ such that $g$ is a sum of squares, 
$$||f-g||^2 \geq \frac{(\tilde{E}[f])^2}{\sum_{I:|I| \leq d}{(\tilde{E}[h_I])^2}} = \frac{\left(c+(r+1)(1-k^r)\right)^2}{(kd^3)^{2d^4}}$$
Taking $k = \sqrt[r]{2+c}$, the right hand side is at least $\frac{1}{(d^3\sqrt[r]{2+c})^{2d^4}}$ and this completes the proof
\end{proof}
\begin{remark}
In fact, we could have taken any polynomial $f$ on a constant number of variables which is non-negative but not a sum of squares and this covers the case of polynomials with degree $0$ mod $4$ (the above construction only gives us polynomials of degree $2$ mod $4$). However, it would be preferable to have an example which really depends on all its variables. We leave this as a question for future work. Also, the constant is a rapidly decaying function of $d$ and it would be very interesting to obtain a more reasonable constant.
\end{remark}
\section{Future Work}
In this paper, we have shown that property testing whether a polynomial is a sum of squares using random samples and a tester which always accepts YES instances is hard; we need $n^{\Omega(d)}$ samples, which is not much less than we would need to completely determine the polnomial. That said, this work raises a number of questions, including but not limited to the following:
\begin{enumerate}
\item What can be shown for adaptive sampling and/or testers which only need to accept YES instances with high probability?
\item What is the threshold at which polynomial interpolation is likely to result in a polynomial with high norm? In other words, what is threshold at which $M$ stops being close to the identity?
\item Can we obtain almost tight bounds on $||g^c||$ for $c \geq 2$ for polynomial interpolation on random points?
\item If a degree $d$ polynomial $f$ is far from being a degree $d$ non-negative polynomial, must it be negative on a constant proportion of inputs?
\item Can we find a degree $d$ polynomial which is non-negative, far from being a sum of squares, and is far from being a junta (even after a change in coordinates)? Can we property test whether there is some basis in which a polynomial $f$ is a junta?
\item For a given $d$, is there a constant $\epsilon \in (0,1)$ where there is a more efficient way to property test whether a polynomial $f$ of norm $1$ is a sum of squares or is $\epsilon$-far from being a sum of squares?
\end{enumerate}
\bibliographystyle{plain}
\bibliography{soda17}

\begin{thebibliography}{10}

\bibitem{bd}
Boaz~Barak \and David~Steurer.
\newblock Sum-of-squares proofs and the quest toward optimal algorithms.
\newblock 2014.

\bibitem{AroraBS10}
Sanjeev Arora, Boaz Barak, and David Steurer.
\newblock Subexponential algorithms for unique games and related problems.
\newblock In {\em FOCS}, pages 563--572, 2010.

\bibitem{AroraRV09}
Sanjeev Arora, Satish Rao, and Umesh~V. Vazirani.
\newblock Expander flows, geometric embeddings and graph partitioning.
\newblock {\em J. {ACM}}, 56(2):5:1--5:37, 2009.

\bibitem{LectureB}
Boaz Barak.
\newblock Lower bounds — 3sat/3xor and planted clique.
\newblock SOS Lecture 3.

\bibitem{Blaisjuntas}
Eric Blais.
\newblock Testing juntas nearly optimally.
\newblock In {\em Proceedings of the Forty-first Annual ACM Symposium on Theory
  of Computing}, STOC '09, pages 151--158, New York, NY, USA, 2009. ACM.

\bibitem{Blekherman}
Grigoriy Blekherman.
\newblock There are significantly more nonnegative polynomials than sums of
  squares.
\newblock {\em Israel Journal of Mathematics}, 153:355–380, 2006.

\bibitem{GoemansW95}
Michel~X. Goemans and David~P. Williamson.
\newblock Improved approximation algorithms for maximum cut and satisfiability
  problems using semidefinite programming.
\newblock {\em J. {ACM}}, 42(6):1115--1145, 1995.

\bibitem{Goldreichbook}
Oded Goldreich.
\newblock Introduction to property testing.
\newblock Cambridge University Press, 2017.

\bibitem{monotonicity}
Oded Goldreich, Shafi Goldwasser, Eric Lehman, Dana Ron, and Alex
  Samorodnitsky.
\newblock Testing monotonicity.
\newblock {\em Combinatorica}, 20(3):301--337, 2000.

\bibitem{Grigoriev01b}
Dima Grigoriev.
\newblock Linear lower bound on degrees of positivstellensatz calculus proofs
  for the parity.
\newblock {\em Theor. Comput. Sci.}, 259(1-2):613--622, 2001.

\bibitem{Hilbert}
David Hilbert.
\newblock Uber die darstellung definiter formen als summe von formen-quadraten.
\newblock {\em Annals of Mathematics}, 32:342--350, 1888.

\bibitem{Lasserre01}
Jean~B. Lasserre.
\newblock Global optimization with polynomials and the problem of moments.
\newblock {\em {SIAM} Journal on Optimization}, 11(3):796--817, 2001.

\bibitem{Motzkin67}
Theodore Motzkin.
\newblock The arithmetic-geometric inequality.
\newblock In {\em Proc. Symposium on Inequalities}, pages 205--224, 1967.

\bibitem{nesterov}
Yurii Nesterov.
\newblock Squared functional systems and optimization problems.
\newblock In {\em High Performance Optimization. Applied Optimization},
  volume~33, pages 405--440, 2000.

\bibitem{LectureA}
Ryan O'Donnell.
\newblock Analysis of boolean functions: Lecture 16: The hypercontractivity
  theorem.
\newblock CMU 18-859S, 2007.

\bibitem{ParnasRS02}
Michal Parnas, Dana Ron, and Alex Samorodnitsky.
\newblock Testing basic boolean formulae.
\newblock {\em {SIAM} J. Discrete Math.}, 16(1):20--46, 2002.

\bibitem{Parrilo00}
Pablo Parrilo.
\newblock Structured semidefinite programs and semialgebraic geometry methods
  in robustness and optimization, 2000.
\newblock PhD thesis, California Institute of Technology.

\bibitem{Schoenebeck08}
Grant Schoenebeck.
\newblock Linear level {L}asserre lower bounds for certain k-{CSP}s.
\newblock In {\em FOCS}, pages 593--602, 2008.

\bibitem{Shor87}
Naum Shor.
\newblock An approach to obtaining global extremums in polynomial mathematical
  programming problems.
\newblock {\em Cybernetics}, 23(5):695–700, 1987.

\end{thebibliography}
\begin{appendix}
\section{Example: 4-XOR polynomial }
In this section, we briefly discuss an attempt at creating a polynomial based on 4-XOR which is far from being non-negative yet passes the property test and why it does not quite work. The polynomial is constructed as follows.
\begin{enumerate}
\item Randomly choose $n^{2 - \delta}$ equations of the form $x_I = b_I$ where $I$ consists of 4 distinct elements of $[1,n]$ and $b_I \in \{-1,+1\}$.
\item Take the polynomial $p = \sum_{I}{-{b_I}x_I}$
\end{enumerate}
As shown by Grigoriev \cite{Grigoriev01b}, later rediscovered by Schoenebeck \cite{Schoenebeck08}, and explained in Boaz Barak's lecture notes \cite{LectureB}, we can construct pseudo-expectation values for a constant $d \geq 4$ as follows:
\begin{enumerate}
\item Start with the equations $x_I = b_I$ for every $I$ which was chosen.
\item As long as there are sets $I,J$ of size at most $d$ such that $|I \Delta J| \leq d$ and we have not yet set $b_{I \Delta J} = {b_I}b_J$, set $b_{I \Delta J} = {b_I}b_J$. If this gives a contradiction because we already set $b_{I \Delta J} = -{b_I}b_J$, halt and fail. However, with high probability this will not happen.
\item Once we are done, we define $\tilde{E}$ as follows.
\begin{enumerate}
\item For all sets $I$ of size at most $d$, take $\tilde{E}[x_I] = b_I$ if $b_I$ was set and take $\tilde{E}[x_I] = 0$ otherwise. 
\item For all sets $I$ of size between $d+1$ and $2d$, we set $\tilde{E}[x_I] = 0$.
\item For all multisets $I$ such that $|I| \leq 2d$ and $I$ contains some $x_i$ with multiplicity 2, we take $\tilde{E}[x_I] = \tilde{E}[x_{I \setminus \{x_i,x_i\}}]$
\end{enumerate}
\end{enumerate}
Observe that $\tilde{E}[p] = -n^{2-\delta}$, so $p$ is very far from being a sum of squares. In fact, using Lemma \ref{distlowerbound}, it can be shown that with high probability, $p$ is $(1 - o(1))||p||$-far from being a sum of squares. If we could add a sum of squares polynomial $g$ to $p$ so that $||g + p||$ is $O(||p||)$ and with high probability polynomially many random samples of $g+p$ will all have nonnegative values, then this would give another example of a polynomial which passes our property tester with high probability yet is far from being a sum of squares. However, this may not be possible. As a special case, if we try taking $g = C$ for a constant $C$ then we would need $C >> ||p||$ in order to make it so that with high probability, polynomially many random samples of $C + p$ all have nonnegative values.
\end{appendix}
\end{document}